\documentclass[sigconf]{acmart}

\usepackage{balance}
\usepackage{microtype}
\usepackage{graphicx}
\usepackage{booktabs} 
\usepackage{bbm}
\usepackage{xfrac}
\usepackage{amsthm}
\usepackage{commath}
\usepackage{amsmath}
\usepackage{epstopdf}
\usepackage{subcaption}
\usepackage{multirow}
\usepackage{xcolor}
\usepackage{color}

\def\BibTeX{{\rm B\kern-.05em{\sc i\kern-.025em b}\kern-.08emT\kern-.1667em\lower.7ex\hbox{E}\kern-.125emX}}
    
\copyrightyear{2019}
\acmYear{2019}
\setcopyright{acmcopyright}
\acmConference[SIGIR '19]{Proceedings of the 42nd International ACM SIGIR
Conference on Research and Development in Information Retrieval}{July 21--25,
2019}{Paris, France}
\acmBooktitle{Proceedings of the 42nd International ACM SIGIR Conference on
Research and Development in Information Retrieval (SIGIR '19), July 21--25, 2019,
Paris, France}
\acmPrice{15.00}
\acmDOI{10.1145/3331184.3331238}
\acmISBN{978-1-4503-6172-9/19/07}

\acmSubmissionID{fp157}

\graphicspath{ {fig/} }

\DeclareMathOperator{\rel}{rel}
\DeclareMathOperator{\rank}{\mathrm{rk}}
\DeclareMathOperator*{\argmax}{argmax}

\newcommand{\algname}{{\sc AllPairs}}
\newcommand{\swapprob}{q}

\newtheorem{prop}{Proposition}

\begin{document}

\title{Intervention Harvesting for Context-Dependent Examination-Bias Estimation}

\author{Zhichong Fang}
\affiliation{
    Tsinghua University\\
    Beijing, China}
\email{fzc14@mails.tsinghua.edu.cn} 
    
\author{Aman Agarwal}
\affiliation{
    Cornell University\\   
    Ithaca, NY, USA}
\email{aa2398@cornell.edu}

\author{Thorsten Joachims}
\affiliation{
    Cornell University\\
    Ithaca, NY, USA}
\email{tj@cs.cornell.edu} 

%

%
\begin{abstract}
 Accurate estimates of examination bias are crucial for unbiased learning-to-rank from implicit feedback in search engines and recommender systems, since they enable the use of Inverse Propensity Score (IPS) weighting techniques to address selection biases and missing data. 
Unfortunately, existing examination-bias estimators are limited to the Position-Based Model (PBM), where the examination bias may only depend on the rank of the document. 
To overcome this limitation, we propose a Contextual Position-Based Model (CPBM) where the examination bias may also depend on a context vector describing the query and the user. 
Furthermore, we propose an effective estimator for the CPBM based on intervention harvesting. 
A key feature of the estimator is that it does not require disruptive interventions but merely exploits natural variation resulting from the use of multiple historic ranking functions. 
Real-world experiments on the ArXiv search engine and semi-synthetic experiments on the Yahoo Learning-To-Rank dataset demonstrate the superior effectiveness and robustness of the new approach. 
 
\end{abstract}

%
%

\ccsdesc[500]{Information systems~Learning to rank}

%
\keywords{examination bias; unbiased learning-to-rank; propensity estimation}

%

%
\maketitle

\section{Introduction}
While implicit feedback (e.g., clicks, dwell time) is an abundant and attractive source of data in most information-retrieval applications (e.g., personal search, email search, recommendation), its use for learning-to-rank (LTR) is challenging due to its biased nature. To address this bias problem, \citet{Joachims/etal/17a} proposed a counterfactual inference approach, providing an unbiased LTR framework via Empirical Risk Minimization. A key requirement for the effectiveness of this approach is an accurate estimate of the examination bias, which describes how likely a user is to discover a particular result. For example, a result is less likely to be discovered at position $10$ than at position $1$. Estimates of the examination bias enable the use of Inverse Propensity Score (IPS) weighting techniques, which make modeling and estimating examination bias equivalent to propensity estimation for unbiased LTR.

There are two key limitations of existing propensity estimation methods for LTR \citep{Agarwal/etal/19, wang2018position, Joachims/etal/17a}. First, existing methods are restricted to the Position-Based Model (PBM) \citep{chuklin2015click}, which only models how examination changes with the rank of the result. Second, existing methods treat all queries uniformly, even though the examination bias is likely to vary from query to query.
For example, users may examine results in navigational queries (i.e., search queries entered with the intention of finding a particular website or webpage) differently compared to informational queries (i.e., search queries for a broad topic for which there could be thousands of relevant results).
To overcome these limitations, a naive approach would be to train a separate PBM for each context -- say one for navigational and one for informational queries -- simply by partitioning the data. However, this is feasible only when there is a small number of discrete contexts, and it does not apply to cases where contexts are described by arbitrary feature vectors. The latter is a highly desirable use case, since it is natural to represent the context by features describing the query (e.g., query length), features describing the candidate set (e.g., size), and features describing the user (e.g., age).

In this paper, we address these limitations of the PBM and present a new Contextual Position-Based Model (CPBM) that greatly extends the expressiveness of the PBM. Instead of having a single examination parameter for each rank that is shared among all queries, we show how the CPBM can model examination dependent on arbitrary context vectors through a deep network. Furthermore, we present an \algname\ estimator \citep{Agarwal/etal/19} for learning CPBM models from log data. For training, our estimator harvests implicit interventions that are already available in most operational systems. In particular, the estimator only requires (not necessarily randomized) log data from at least two ranking functions that were deployed on the system in the past. The resulting deep network can then be used to compute context-dependent propensities for LTR algorithms like \citep{Joachims/etal/17a,Agarwal/etal/18b,Agarwal/etal/19b}. We evaluate the fidelity of the CPBM model and the effectiveness of the estimator in real-world experiments on the ArXiv full-text search engine and in semi-synthetic experiments on the Yahoo Learning-to-Rank Challenge dataset \cite{chapelle2011yahoo}.

\section{Related Work}
In most information retrieval systems, large amounts of implicit feedback are logged automatically and serve as an attractive source of training data.
However, it is known that this type of data suffers from various biases due to both the system and the user, such as position bias \citep{Joachims/etal/05a}, presentation bias \citep{o2006modeling} and trust bias \citep{joachims2007evaluating}. 

To handle biases in a principled way, \citet{Joachims/etal/17a} introduced an unbiased learning-to-rank framework, which is a consistent learning approach despite biased feedback.
It relies on IPS weighting first developed in causal inference and survey sampling \citep{Horvitz1952,rosenbaum1983central}.
IPS has been commonly adopted for unbiased evaluation and learning \citep{Agarwal/etal/18b, dudik2011doubly, li2015counterfactual, Swaminathan/Joachims/15c,Agarwal/etal/19b}.
However, because the propensity in the unbiased LTR setting represents the unknown bias with which a user examines a document, this propensity needs to be estimated.

Existing propensity-estimation methods for LTR are based on the Position-Based Model (PBM) \citep{richardson2007predicting}. 
The most effective methods use randomized interventions \citep{wang2016learning, Joachims/etal/17a}, which unfortunately degrade the user's search experience.
To avoid such interventions, \citet{wang2018position} proposed a regression-based Expectation-Maximization (EM) algorithm, and \citet{ai2018unbiased} proposed a learning algorithm that learns propensity models together with the ranking model. 
Unfortunately, both approaches involve learning an accurate relevance model, which is just as hard as the LTR problem itself.
The approach of \citet{Agarwal/etal/19} avoids both randomized interventions and relevance modeling by exploiting click data from multiple loggers as implicit interventions.
In our work, we extend their approach to the Contextual Position-Based Model (CPBM) for improved accuracy.
 
Beyond the PBM, many other click models for ranked search exist. However, they were designed for inferring relevance, not propensities.
One example is the Cascade model \citep{craswell2008experimental}, where users scan documents top-down until a relevant document is found.
Built upon the PBM and the Cascade model, more complex models like UBM \citep{dupret2008user}, DBN \citep{chapelle2009dynamic}, CCM \citep{guo2009click} and CSM \citep{Borisov:2018:CSM:3209978.3210004} were proposed to infer relevance judgments from click logs. It is an open question in how far these models can be adapted for propensity estimation as well.

\section{The Contextual Position-Based Model}

Modeling the examination bias is crucial for learning to rank from implicit feedback, since it confounds the feedback signal. We start by reviewing the Position-Based Model, as it is arguably the simplest model for correcting the examination bias in learning to rank from implicit feedback. As shown by \citet{Joachims/etal/17a}, the parameters of the PBM can serve as propensity estimates, enabling the use of IPS weighting for unbiased learning-to-rank.

The PBM captures that the rank of a result has a strong influence on whether a result is examined (i.e. viewed and evaluated as a prerequisite for any subsequent feedback like a click or a rating) by a user, where higher-ranked results are typically more likely to be examined than results further down the ranking. Suppose that for a particular query $q$, result $d$ is displayed at position $k$. Let $C$ be the random variable corresponding to a user clicking on $d$, and let $E$ be the random variable denoting whether the user examines $d$. Then according to the Position-Based Model \citep{chuklin2015click},
\begin{equation}
    \Pr(C=1|q,d,k) =  \Pr(E=1|k) \rel(q,d),
\end{equation}
where $\rel(q,d) \in \{0,1\}$ is the binary relevance of document $d$ for query $q$. 

While $\Pr(E=1|k)$ can be used as an estimate of the examination propensity \citep{Joachims/etal/17a}, it is a rather simplistic model since it assumes that examination does not vary across queries. However, it is implausible that navigational queries share the same propensity curve with informational queries, and we will validate in our experiments that such dependencies exist in real-world search engines. More broadly, we argue that examination behavior not only varies across queries, but that it varies across contexts $x$ more generally. This context $x$ includes the query itself and features describing the query (e.g., query length), features describing the candidate set (e.g., size), and features describing the user (e.g., age). 
To be able to model these dependencies, we propose a new model -- called the Contextual PBM (CPBM) -- where the examination propensity can depend on the observed context $x$ in addition to the position as follows.
\begin{equation}
    \Pr(C=1|x,d,k) =  \Pr(E=1|k,x) \rel(x,d).
\end{equation}
Since the context $x$ contains all the information about its corresponding query $q$, we can drop the query $q$ from our notation. 
Through its dependence on context $x$, the CPBM can represent different propensity curves $\Pr(E=1|k,x)$ w.r.t. position $k$ for each query context $x$, instead of assuming that all queries share the same examination curve $\Pr(E=1|k)$ like in the PBM.

\section{Estimating CPBM Models} \label{sec:est}
While the increased expressiveness of the CPBM is clearly desirable, it raises several challenges when estimating the model from the data. In particular, instead of just estimating $k_{max}$ scalar parameters $\Pr(E=1|k)$ like in the PBM, where $k_{max}$ is the maximum length of the presented rankings (say $10$ or $20$), the CPBM requires estimating a context-dependent propensity model $\Pr(E=1|k,x)$, which in the following will be represented as a neural network. Furthermore, estimating $\Pr(E=1|k,x)$ is challenging since we typically do not observe ground truth for $\rel(x,d)$ such that it is difficult to attribute the lack of a feedback signal to a lack of examination or a lack of relevance. After reviewing the shortcomings of a naive generative modeling approach in the next subsection, we will exploit the fact that randomized interventions can be used to control for relevance. In particular, we will show how reusing logged click data from multiple ranking functions provides such intervention data for the CPBM under reasonable assumptions, eliminating the need for explicit interventions that affect the user experience. 


\begin{figure*}[tb]
  \centering
  \begin{subfigure}[t]{0.65\columnwidth}
    \centering
    \includegraphics[height=4cm]{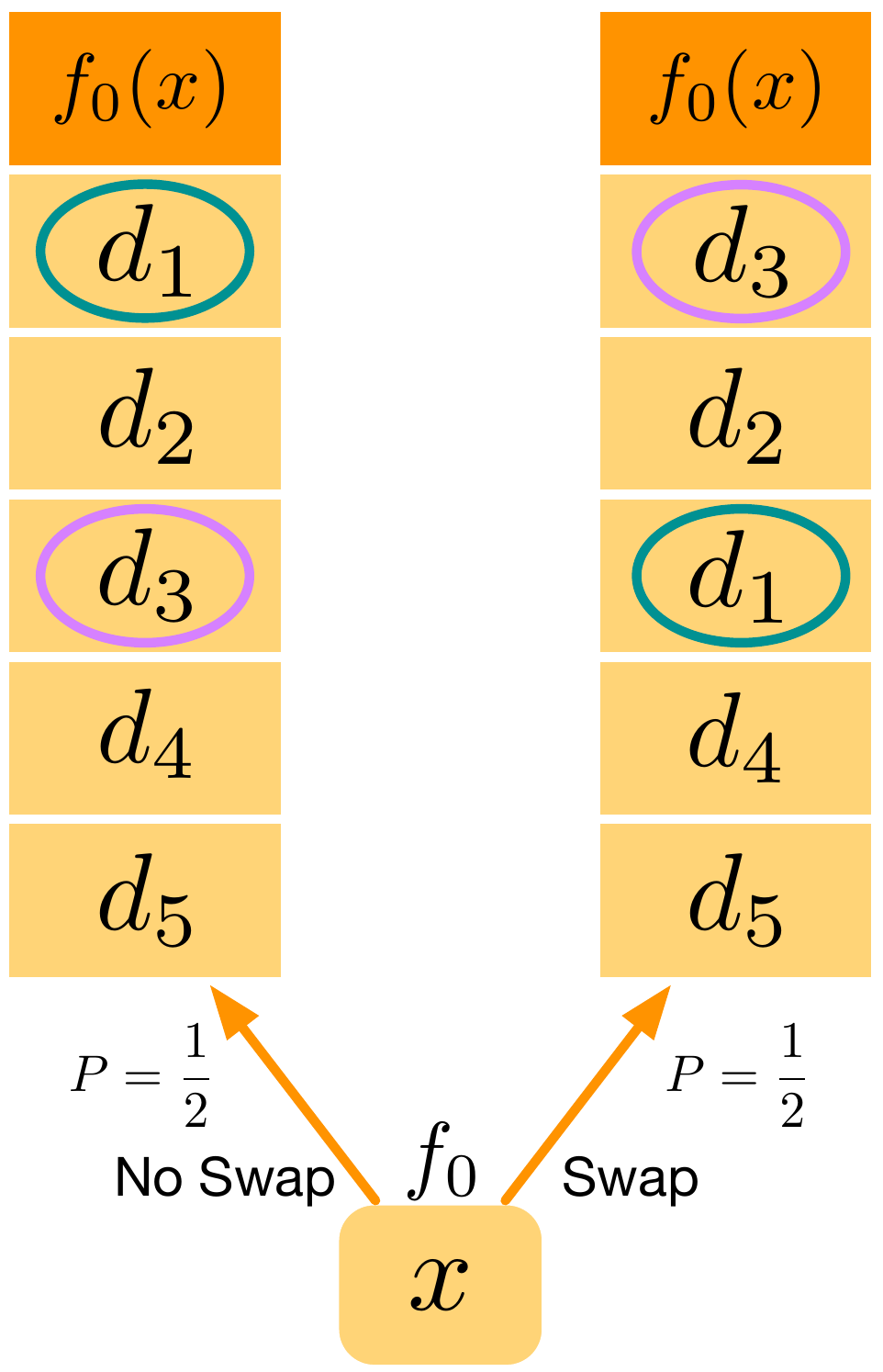}
    \caption{Swap Intervention between positions $1$ and $3$.}
    \label{fig:swap_toy}
  \end{subfigure}
  \begin{subfigure}[t]{0.65\columnwidth}
    \centering
    \includegraphics[height=4cm]{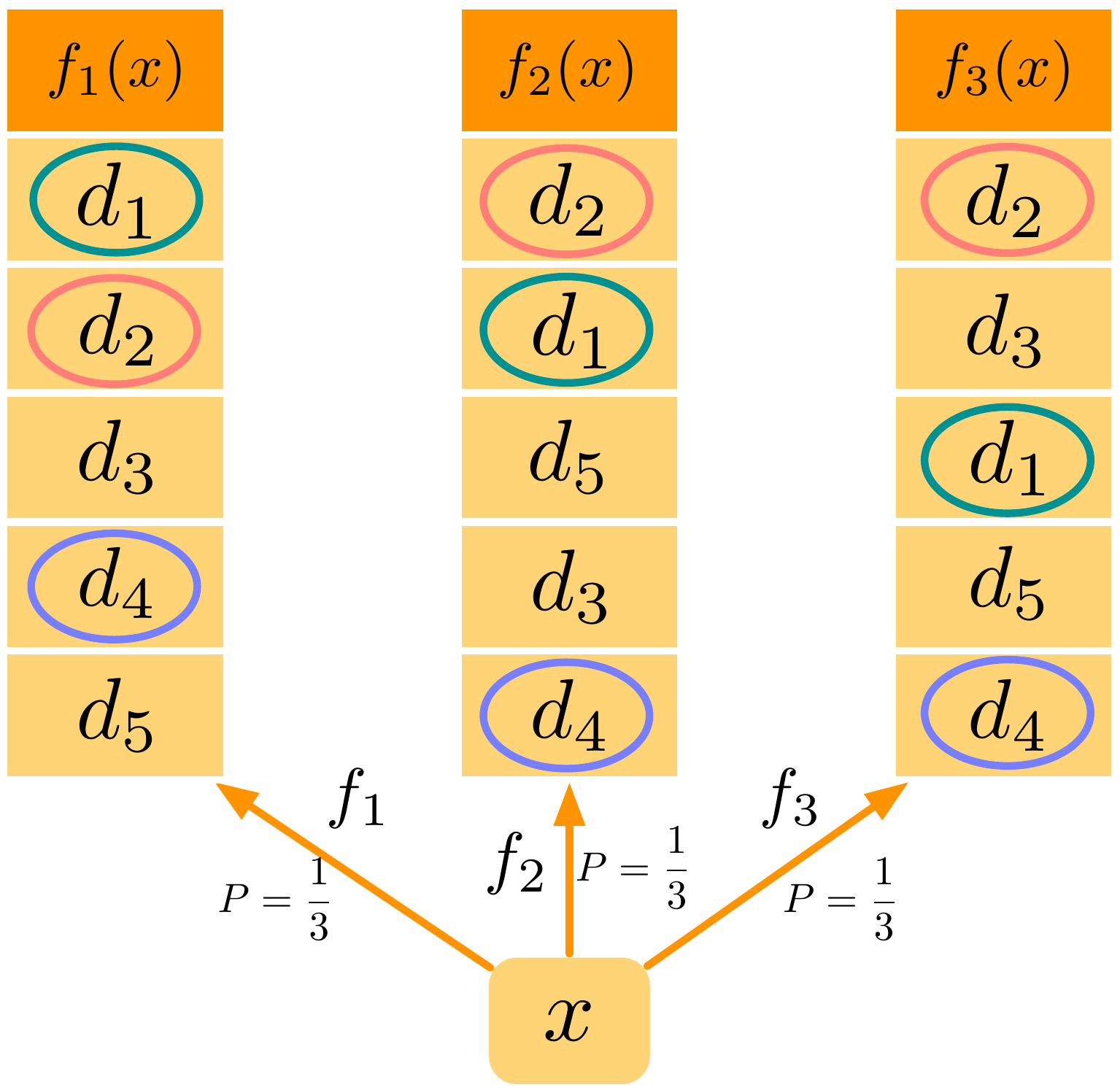}
    \caption{Intervention Harvesting (A/B Test).}
    \label{fig:ab_toy}
  \end{subfigure}
  \begin{subfigure}[t]{0.65\columnwidth}
    \centering
    \includegraphics[height=4cm]{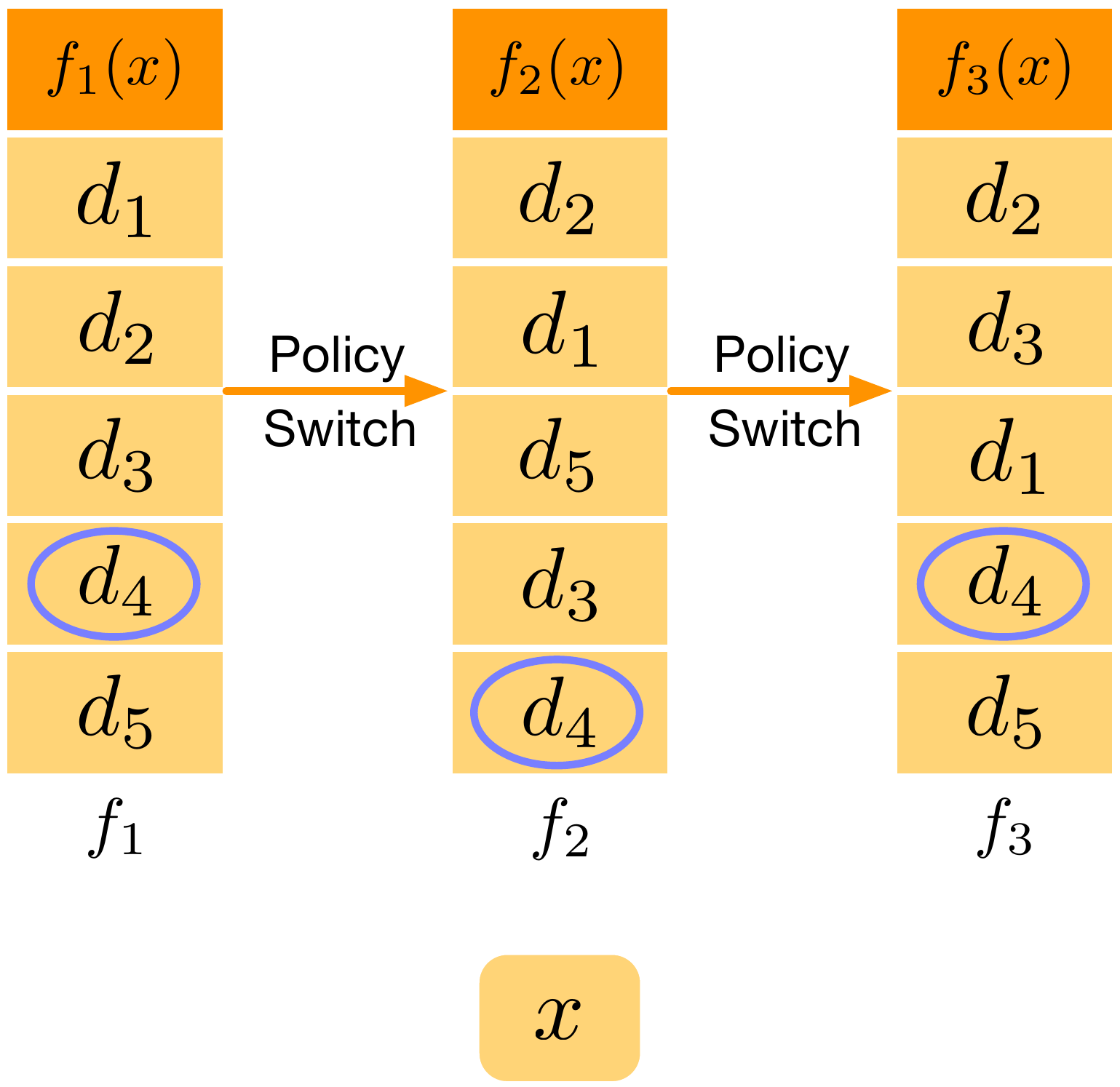}
    \caption{Intervention Harvesting (switch in production rankers).}
    \label{fig:policy_toy}
  \end{subfigure}
 \vspace*{-0.3cm}
 \caption{Illustration of Swap Interventions and of Intervention Harvesting.}
  \label{fig:toy}
\end{figure*}

\subsection{Generative Modeling} 

The first thought one may have is to estimate a CPBM via a standard generative-modeling approach with both examination and relevance as latent variables. In fact, \citet{wang2018position} have proposed such an approach for the simpler problem of estimating the parameters of the PBM. Let $\mathcal{L}=\{(x^j, d^j, k^j, c^j)|j\in[N]\}$ be a sample of $N$ observations with one tuple for each context-document pair $(x^j,d^j)$, indicating with $k^j$ the position of $d^j$ in the ranking and with $c^j \in \{0,1\}$ whether it was clicked. 
Extending the approach of \citet{wang2018position} to the CPBM, the conditional log likelihood objective (conditioned on the observed queries and rankings) is
\begin{align*}
\log \Pr(\mathcal{L}) =  \sum_{j\in\mathcal{L}} \bigg[ & c^j \log\left( p_{k^j}(x^j)r(x^j,d^j)\right)  \\[-0.3cm] &+ (1-c^j) \log\left(1-p_{k^j}(x^j) r(x^j,d^j)\right)\bigg],
\vspace*{-0.1cm}
\end{align*}
where $p_k(x):=\Pr(E=1|k,x)$ is a context-dependent propensity model and $r(x,d):=\rel(x,d)$ is a document-dependent relevance model. Both relevance and examination are latent, and even for the simpler PBM model it was found that the propensity estimates can be far off \cite{Agarwal/etal/19}. A key shortcoming of this approach is that it requires learning the relevance $\rel(x,d)$ of all individual documents without any direct supervision, which is just as difficult as the learning-to-rank problem itself. This means that the relevance model will typically be misspecified and thus bias the propensity estimates.


\subsection{Explicit Swap Interventions} 
\label{sec:swap}

To overcome the need for modeling the unobserved relevance of all query-document pairs, we will employ an interventional approach that controls for relevance at each position. To start, let us first review how explicit interventions have been used for estimating $p_k := \Pr(E=1|k)$ in the PBM \cite{Joachims/etal/17a,wang2016learning}. The PBM requires estimating a single vector $p = [p_1, p_2, ..., p_{k_{max}}]$ with $p_k$ for each position $k \in [1,k_{max}]$. 
In this case, randomly swapping results at positions $k$ and $k'$ before presenting the ranking \citep{Joachims/etal/17a} makes the expected relevance of results at the two positions equal. An illustrative example is given in Figure~\ref{fig:swap_toy} for $k=1$ and $k'=3$.
Through the randomized swap, documents $d_1$ and $d_3$ have a 50\% probability of being presented either at position $k=1$ or at position $k'=3$. So, over the distribution of queries that are subject to this randomized swapping, the distribution of documents in position $k=1$ is identical to the distribution of documents in $k=3$, and thus is their expected relevance. This randomized control for relevance resolves the ambiguity in attributing the lack of clicks to either a lack of relevance or a lack of observation.

More formally, denote with $C_{k,k'}^k$ and $C_{k,k'}^{k'}$ the random variables indicating clicks on positions $k$ and $k'$ respectively for the set of training queries where the results at positions $k$ and $k'$ are swap-randomized with probability $\swapprob=0.5$. Since the results are swapped uniformly, the expected relevance at positions $k$ and $k'$ is controlled to be equal at these positions, and thus expected click-through rates reveal the relative propensities via
$$\frac{p_k}{p_{k'}} = \frac{\mathrm{E}[C_{k,k'}^k]}{\mathrm{E}[C_{k,k'}^{k'}]}.$$
This means that the ratio of the observed click-through rates is a consistent estimator of the relative propensities $p_k$ and $p_{k'}$ under the PBM \citep{Joachims/etal/17a}. Note that knowing the relative propensities with respect to a single ``anchor" position (e.g. $\frac{p_k}{p_{1}}$) is sufficient, since the counterfactual ERM learning objective is invariant to multiplicative scaling \citep{Joachims/etal/17a}. 

While this ratio estimator is a sensible approach for the PBM, it is not directly applicable to the Contextual PBM even if we only need relative propensity estimates. 
In particular, a simple ratio of the observed click-through rates at different ranks will yield  $\frac{\mathrm{E}_x [p_k(x)]}{\mathrm{E}_x [p_{k'}(x)]}$, where the expectation is over contexts. This is not the estimate we seek for the CPBM, since we need estimates of each specific $p_k(x)$ (up to multiplicative scaling) to de-bias clicked examples at position $k$ under context $x$. 
To get such context-dependent propensity estimates, we will introduce a different estimator below. 
Furthermore, we will show how to avoid explicit swap interventions by harvesting implicit interventions. 
As illustrated below, such implicit interventions are typically available in large quantities and do not come at the expense of user experience related to randomly swapping results. 

\subsection{Intervention Harvesting for the CPBM} 


Instead of explicitly swapping results, \citet{Agarwal/etal/19} have recently shown for the PBM how interventions similar to explicit swaps can be harvested from data that is readily available in most operational systems. We will extend this approach to the CPBM and derive an intervention-harvesting estimator for the CPBM that does not require explicit swap interventions, nor does it require a document-specific relevance model that would be difficult to fit. Instead, our estimator merely needs to model how the {\it average relevance} over all queries and documents at a position -- not the context-document specific relevance -- changes with context.


As input for our estimator, suppose we have data from $m$ historic rankers $F=\{f_1,...,f_m\}$. 
Each ranker $f$ maps a query context $x$ to a ranking $f(x)$ of the candidate set of documents. Let $\rank(d|f(x))$ denote the rank of document $d$ in the ranking.
Let $n_i$ be the number of queries that $f_i$ processed, and let $\mathcal{L}=\{(x^j, d^j, k^j, c^j)|j\in[N]\}$ be the aggregated click log over all the rankers, with one tuple for each context-document pair. 
We require that the distribution of contexts is stationary, or specifically that there is no dependency between the context and the choice of ranking function $f_i$ \citep{Langford:2008:ES:1390156.1390223,Agarwal/etal/19},
\begin{equation}
    \forall f_i: \Pr(\mathcal{X}|f_i)=\Pr(\mathcal{X}) \Rightarrow \forall x \in \mathcal{X}: \Pr(f_i|x) = \Pr(f_i).
    \label{eq:cond}
\end{equation}
This condition is fulfilled in at least two situations -- namely in A/B tests and under stationary $\Pr(\mathcal{X})$. In data from A/B tests, where users are randomly assigned to one of the rankers, the condition is fulfilled by design. 
An example is shown in Figure~\ref{fig:ab_toy}. For a given context $x$, the ranking functions $f_1$, $f_2$ and $f_3$ are each chosen completely randomized with equal probability $\frac{1}{3}$. By choosing one of the three rankers, we implicitly conduct a number of interventions. For example, document $d_1$ is randomized to be displayed in positions $1$, $2$, or $3$ with equal probability, and document $d_2$ is displayed in position $1$ with probability $\frac{2}{3}$ and in position $2$ with probability $\frac{1}{3}$. Figure~\ref{fig:policy_toy} shows that a similar randomization holds when the production ranker gets updated from $f_1$ to $f_2$ under stationary $\Pr(\mathcal{X})$. Stationarity implies that the probability of a context $x$ is equal before and after the update, and thus $d_4$ has twice the probability of being shown in position $4$ than in position $5$ in this toy example with 3-time steps. 

To exploit this readily available intervention data for estimating the CPBM, let's first focus on a fixed pair of positions $k, k'$. The key idea of intervention harvesting for the CPBM is to control for the varying average relevance of results displayed in positions $k,k'$ for context $x$ by restricting to the set of queries that, for an appropriate choice of ranker from $F=\{f_1,...,f_m\}$, could have been placed either at $k$ or $k'$. To this effect, we define interventional sets 
\begin{equation}
    S_{k,k'} := \{(x,d) :\exists f,\!f' \,  \rank(d|f\!(x)) \!=\! k \wedge \rank(d|f'\!(x)) \!=\! k'\}
\end{equation}
as the sets of $(x,d)$ pairs that receive ``treatments" $k$ or $k'$ under different rankers. 
Specifically, a context-document pair $(x,d)$ is included in $S_{k,k'}$, if for the context $x$ some ranker $f \in F$ puts the document $d$ at position $k$ and another ranker $f' \in F$ puts it at position $k'$. This is akin to a virtual swap intervention at positions $k$ and $k'$, albeit only with a single document. 
Based on these definitions, the toy example in Figure~\ref{fig:ab_toy} produces interventional sets such that $(x,d_1) \in S_{1,2}, S_{1,3}, S_{2,3}$, $(x,d_2) \in S_{1,2}$, $(x,d_4) \in S_{4,5}$, etc.
Note that the set includes all possible queries that may be sampled, not only those that are actually sampled in one or more rankers' logs. 
Furthermore, note that the feedback signals of $(x,d)$ from some rankers might remain counterfactual and unobserved.
Illustrating this using the toy example in Figure~\ref{fig:ab_toy}, each ranking of $F = \{f_1, f_2, f_3\}$ was a potential choice, but only one of those rankings was presented to the user -- say the ranking of $f_1$. In this way, we only observe the feedback for $d_1$ at position 1, but not at the other positions.

To account for the fact that not all interventions within an interventional set $S_{k,k'}$ have the same probability, we define the following weighting function that is proportional to the treatment-assignment probability. It can either be computed from the known assignment probabilities in an A/B test, or for consecutive policy deployments via
\begin{equation}
    q_k(x,d) := \frac{\sum_{i=1}^m n_i\mathbbm{1}[\rank(d|f_{i}(x)) = k]}{\sum_{i=1}^m n_i}.
\end{equation}
For the example in Figure~\ref{fig:ab_toy}, we have  $q_1(x,d_1)=q_2(x,d_1)=q_3(x,d_1)=\frac{1}{3}$, $q_4(x,d_4) = \frac{1}{3}$, $q_5(x,d_4) = \frac{2}{3}$, etc.

\subsection{AllPairs Estimator for the CPBM}

Now that we have extracted intervention data and its assignment mechanism from existing logs, we can tackle the question of defining an estimator for the CPBM using this data. The key challenge compared to analogous estimators for the PBM \cite{Agarwal/etal/19} lies in modeling the dependence on context. We start by constructing the following feedback labels for each $(x^j, d^j, k^j, c^j) \in \mathcal{L}$ by correcting the non-uniform assignment mechanism to the uniform intervention distribution in each interventional set \cite{Agarwal/etal/19}.
\begin{align*} 
    \hat{c}^j_{k,k'}(k) &:= \mathbbm{1}_{[(x^j,d^j) \in  S_{k,k'}]}\mathbbm{1}_{[k^j=k]}\frac{c^j}{q_k(x^j,d^j)}\\
    \hat{\neg c}^j_{k,k'}(k)&:= \mathbbm{1}_{[(x^j,d^j) \in S_{k,k'}]}\mathbbm{1}_{[k^j=k]}\frac{1-c^j}{q_k(x^j,d^j)}
\end{align*}
This can be thought of as an IPS weighted class label. 
For the PBM without a dependence on context $x$, in expectation (over the choice of ranker and query) $\hat{c}^j_{k,k'}(k)$ is proportional to the product of examination propensity $p_k$ and average relevance $r_{k,k'}\!:=\!\Pr(\rel(x,d)=1|(x,d)\!\in\!S_{k,k'})$ \citep{Agarwal/etal/19}, just like for the explicit swap interventions mentioned above. 
However, when there is a dependency on context $x$ for both the examination probabilities $p_k(x)$ and the average relevance $r_{k,k'}(x)\!:=\!\Pr(\rel(x,d)=1|(x,d) \in S_{k,k'},x)$, unbiasedness w.r.t.\ the query distribution no longer holds and there is generally no small number of individual parameters $p_k$ and $r_{k,k'}$ that could be estimated exhaustively. 
To overcome this problem, we exploit that unbiasedness still holds for each individual context, and we introduce a context-dependent examination model $h(k,x)$ for $p_k(x)$ and a context-dependent average relevance model $g(k,k',x)$ for $r_{k,k'}(x)$ to compactly capture the variation across contexts. In the experiments in this paper, we model both $h(k,x)$ and $g(k,k',x)$ as neural networks.

With these definitions in place, we can now formulate an extremum estimator similar to a maximum-likelihood criterion. We call this the \algname\ estimator for the CPBM. It combines the flexibility of the observational generative modeling approach with the robustness of the interventional methods, specifically the intervention harvesting approach previously used for estimating the PBM \citep{Agarwal/etal/19}.
\begin{align}
\hat{h}_{CPBM}:= \argmax_{h,g} \!\!\!\sum_{j\in\mathcal{L}}\sum_{k\neq k'}\!\bigg[&\hat{c}_{k,k'}^j\!(k)\log\left(h(k,x^j)g(k,k'\!,x^j)\right) \label{eq:obj}\\[-0.3cm]
& \!+ \! \hat{\neg c}_{k,k'}^j\!(k)\log\left(1\!-\!h(k,x^j)g(k,k'\!,x^j)\right)\!\!\bigg] \nonumber
\end{align}
Here, $h(k,x)$ and $g(k,k',x)$ are constrained to $(0,1)$ by using a sigmoid output layer on both networks. While the \algname\ estimator has syntactic similarity with the generative maximum-likelihood objective from \cite{wang2018position}, both are fundamentally different. Notably, \algname\ uses interventional data to control for the unobserved document relevance, while the generative model is purely observational. This allows the average relevance model $g(k,k',x)$ in \algname\ to be substantially simpler than the individual relevance model $g(q,d)$ in generative modeling. In particular, $g(k,k',x)$ in \algname\ does not model the relevance of an individual document to a query, but merely how the average relevance of documents in positions $k$ and $k'$ changes with context. As such, $g(k,k',x)$ does not require document-level relevance features, but merely takes the context $x$ and the positions as input. In the experiments, we find that the average relevance at a specific position $k$ does not change much with context $x$, and that even replacing the neural relevance model $g(k,k',x)$ with $k_{max}$ choose $2$ context-independent parameters $r_{k,k'}$ performs quite well.

We now further justify the use of the objective in Equation~(\ref{eq:obj}) by showing that it is equivalent to a weighted version of Cross-Entropy Maximization where the weights adjust for the varying amounts of interventional data available across position pairs $k,k'$. This relates the \algname\ objective to optimizing the KL-divergence between model and data, and it implies two practical advantages. First, for this type of objective, it is well known that training neural networks via backpropagation is effective.
Second, this objective provides an attractive method for information aggregation, mitigating the noisiness and sparsity of click data.

\begin{prop}
Under the condition in (\ref{eq:cond}) and i.i.d. contexts $x \sim \Pr(\mathcal{X})$, the objective in Equation~(\ref{eq:obj}) is equivalent to the following weighted form of Cross-Entropy,
\begin{align*}
\sum_{x\in \mathcal{X}}\sum_{k\neq k'}\hat{N}_{k,k'}(x) \bigg[&\hat{y}_{k,k'}(k,x) \log \left(y_{k,k'}(k,x)\right)\\
& \!+ \!\hat{\neg y}_{k,k'}(k,x) \log \left(1\!-\!y_{k,k'}(k,x)\right)\bigg]
\end{align*}
of the random variables $y_{k,k'}(k,x)$ and their empirical counterparts $\hat{y}_{k,k'}(k,x)$, $\hat{\neg y}_{k,k'}(k,x)$ weighted with $\hat{N}_{k,k'}(x)$, where 
\begin{align*}
 & y_{k,k'}(k,x) := h(k,x)g(k,k',x) = p_k(x) r_{k,k'}(x) \\
 & \hat{N}_{k,k'}(x) :=\!\!\sum_{j\in\mathcal{L}} \!\!\mathbbm{1}_{[x^{j}=x]} \mathbbm{1}_{[(x^j,d^j) \in  S_{k,k'}]} \\
 & \hat{y}_{k,k'}(k,x) :=\!\!\frac{\sum_{j\in\mathcal{L}} \!\!\mathbbm{1}_{[x^{j}=x]}\hat{c}_{k,k'}^j\!\!(k)}{\hat{N}_{k,k'}(x)} \\ 
 & \hat{\neg y}_{k,k'}(k,x) :=\!\!\frac{\sum_{j\in\mathcal{L}} \!\!\mathbbm{1}_{[x^{j}=x]}\hat{\neg c}_{k,k'}^j\!\!(k)}{\hat{N}_{k,k'}(x)},
\end{align*}
and 
\begin{align*}
&\mathbbm{E}[\hat{y}_{k,k'}(k,x)]=y_{k,k'}(k,x)\\
&\mathbbm{E}[\hat{\neg y}_{k,k'}(k,x)]=1-y_{k,k'}(k,x).
\end{align*}
\end{prop}

\begin{proof}
First, we rewrite the objective as follows:
\begin{align*}
\sum_{j\in\mathcal{L}}&\sum_{k\neq k'} \!\bigg[\hat{c}_{k,k'}^j(k)\log \left(y_{k,k'}(k,x^j)\right)
+ \hat{\neg c}_{k,k'}^j\!(k)\log \left(1\!-\!y_{k,k'}(k,x^j)\right)\!\bigg]\\
=& \sum_{x\in \mathcal{X}} \sum_{j\in\mathcal{L}}\sum_{k\neq k'} \mathbbm{1}_{[x^{j}=x]} \!\bigg[\hat{c}_{k,k'}^j(k)\log\left(y_{k,k'}(k,x^j)\right)\\
& +\hat{\neg c}_{k,k'}^j(k)\log\left(1\!-\!y_{k,k'}(k,x^j)\right)\bigg]\\
\end{align*}
\begin{align*}
=& \sum_{x\in \mathcal{X}}\sum_{k\neq k'}\!\bigg[\left(\sum_{j\in\mathcal{L}} \mathbbm{1}_{[x^{j}=x]}\hat{c}_{k,k'}^j(k)\right)\log\left(y_{k,k'}(k,x^j)\right)\\
& + \left( \sum_{j\in\mathcal{L}} \mathbbm{1}_{[x^{j}=x]}\hat{\neg c}_{k,k'}^j(k)\right)\log \left(1\!-\!y_{k,k'}(k,x^j)\right)\bigg]\\
=& \sum_{x\in \mathcal{X}}\!\sum_{k\neq k'}\! \hat{N}_{k,k'}(x)\bigg[ \hat{y}_{k,k'}(k,x)\log\left(y_{k,k'}(k,x^j)\right)\\
&+ \hat{\neg y}_{k,k'}(k,x)\log\left(1\!-\!y_{k,k'}(k,x^j)\right)\bigg]
\end{align*}
Next, we are going to prove that $\mathbbm{E}[\hat{y}_{k,k'}(k,x)]=y_{k,k'}(k,x)$ and $\mathbbm{E}[\hat{\neg y}_{k,k'}(k,x)]=1-y_{k,k'}(k,x)$, which is required by Cross-Entropy. 
\begin{align*}
&\mathbbm{E}[\hat{y}_{k,k'}(k,x)\hat{N}_{k,k'}(x)]\\
=&\mathbbm{E}[\sum_{j\in\mathcal{L}} \mathbbm{1}_{[x^{j}=x]}\mathbbm{1}_{[(x^j,d^j) \in  S_{k,k'}]}\mathbbm{1}_{[k^j=k]}\frac{c^j}{q_{k}(x^j,d^j)}]\\
=&\sum_{i=1}^m \sum_{j=1}^{n_i}\!\Pr(x)\!\!\!\!\!\!\sum_{d \in \Omega(x)}\!\!\!\!\mathbbm{1}_{[(x,d) \in  S_{k,k'}]}\mathbbm{1}_{[\mathrm{rk}(d|f_i(x))=k]}\frac{\mathbbm{E}[c(d)]}{q_{k}(x,d)} \\
=&\Pr(x)\!\!\sum_{i=1}^m \sum_{j=1}^{n_i}  \!\sum_{d \in \Omega(x)}\!\!\!\!\mathbbm{1}_{[(x,d) \in  S_{k,k'}]}\mathbbm{1}_{[\mathrm{rk}(d|f_i(x))=k]}\frac{p_k(x)\rel(x,d)}{q_k(x,d)} \\
=&p_k(x)\Pr(x)\!\!\!\!\!\!\sum_{d \in \Omega(x)}\!\!\!\!\mathbbm{1}_{[(x,d) \in  S_{k,k'}]}\rel(x,d)\frac{\sum_{i=1}^m \!\!\sum_{j=1}^{n_i}\!\!\mathbbm{1}_{[\mathrm{rk}(d|f_i(x))=k]}}{q_k(x,d)}\\
=&p_k(x)\Pr(x)\!\!\!\!\sum_{d \in \Omega(x)}\!\!\!\!\mathbbm{1}_{[(x,d) \in  S_{k,k'}]}\rel(x,d)\frac{\sum_{i=1}^m \!\!n_i\!\mathbbm{1}_{[\mathrm{rk}(d|f_i(x))=k]}}{q_k(x,d)}\\
=&p_k(x)\Pr(x)\!\!\!\!\sum_{d \in \Omega(x)}\!\!\!\!\mathbbm{1}_{[(x,d) \in  S_{k,k'}]}\rel(x,d) \sum_{i=1}^{m}n_i\\
=&p_k(x)\mathbbm{E}[\sum_{j\in\mathcal{L}} \mathbbm{1}_{[x^{j}=x]}\mathbbm{1}_{[(x^j,d^j) \in  S_{k,k'}]}\rel(x^j,d^j)]\\
=&p_k(x) r_{k,k'}(x) N_{k,k'}(x) \\
=&y_{k,k'}(k,x)N_{k,k'}(x) 
\end{align*}
where $N_{k,k'}(x) = \mathbbm{E}[\sum_{j\in\mathcal{L}} \mathbbm{1}_{[x^{j}=x]}\mathbbm{1}_{[(x^j,d^j) \in  S_{k,k'}]}] = \mathbbm{E}[\hat{N}_{k,k'}(x)]$.
Then we have $\mathbbm{E}[\hat{y}_{k,k'}(k,x)] = \frac{\mathbbm{E}[\hat{y}_{k,k'}(k,x)\hat{N}_{k,k'}(x)]}{\mathbbm{E}[\hat{N}_{k,k'}(x)]} = y_{k,k'}(k,x)$.
Similarly, $\mathbbm{E}[\hat{\neg y}_{k,k'}(k,x)] = 1 - y_{k,k'}(k,x)$. 
Note that we make the reasonable assumption that user click behavior is independent of the context sampling and ranker choice process, and thus $\hat{y}_{k,k'}(k,x)$ and $\hat{N}_{k,k'}(x)$ are independent random variables, so that $\mathbbm{E}[\hat{y}_{k,k'}(k,x)\hat{N}_{k,k'}(x)] = \mathbbm{E}[\hat{y}_{k,k'}(k,x)]\mathbbm{E}[\hat{N}_{k,k'}(x)]$.

\end{proof}

\begin{figure}[tb]
  \centering
  \begin{subfigure}[tb]{\columnwidth}
    \centering
    \includegraphics[height=3.5cm]{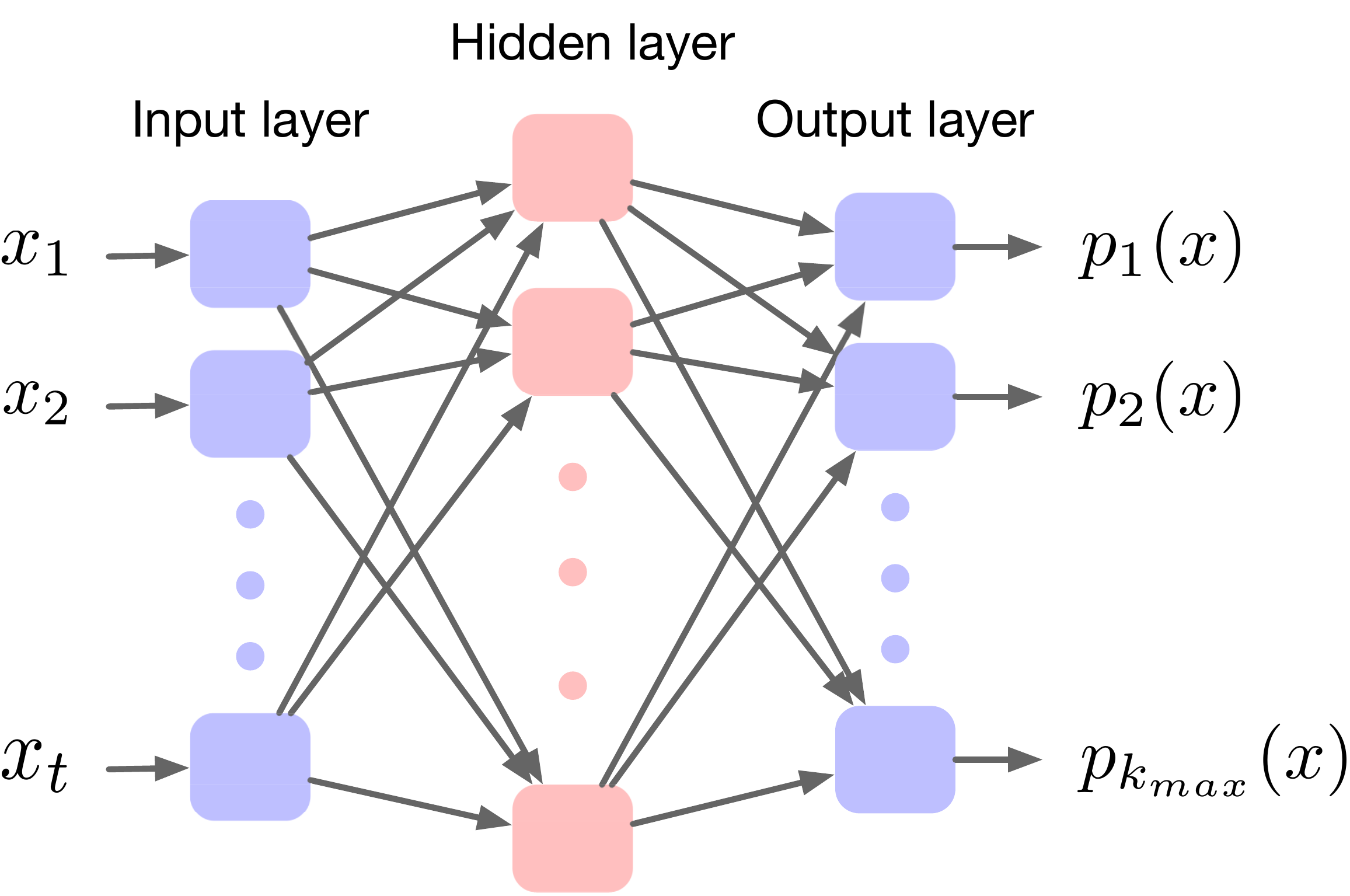}
    \caption{Propensity Model.}
    \label{fig:prop_mlp}
  \end{subfigure}\\
  \begin{subfigure}[tb]{\columnwidth}
    \centering
    \includegraphics[height=3.5cm]{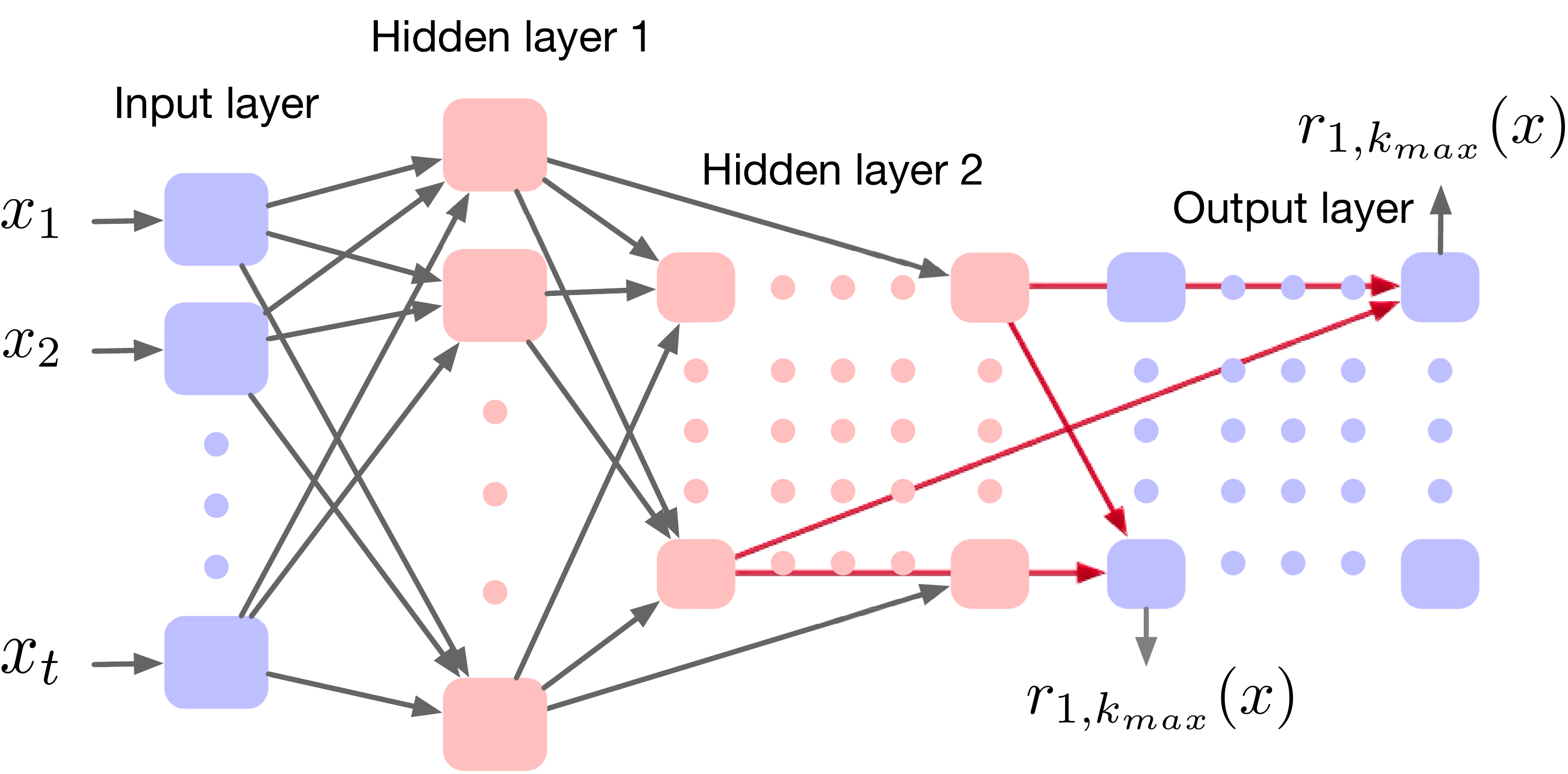}
    \caption{Relevance Model.}
    \label{fig:rel_mlp}
  \end{subfigure}
    \caption{The architecture of the multilayer perceptrons.}
  \label{fig:mlp}
\end{figure}

\subsection{Neural Network Model for the CPBM} \label{sec:mlp}

We employ neural networks for modeling both the context-depen\-dent propensities $h(k,x)$ as well as the context-dependent average relevance $g(k,k',x)$. The respective multi-layer perceptron (MLP) architectures are shown in Figure~\ref{fig:mlp}. Both networks take as input the context features $x \in \mathbb{R}^t$, and output the examination propensity vector $p(x) \in \mathbb{R}^{k_{max}} $ and the average relevance matrix $r(x) \in \mathbb{R}^{k_{max} \times k_{max}}$ respectively. 
The hidden layer in the propensity MLP is a traditional sigmoid-activated dense layer, which learns a weight matrix $\mathrm{W_p} \in \mathbb{R}^{t \times k_{max}}$ and bias vector $\mathrm{b_p} \in \mathbb{R}^{k_{max}}$ to produce a propensity vector $p(x) = \sigma(\mathrm{W_p}x+\mathrm{b_p})$. The average-relevance model is less standard, and its first hidden layer in the relevance MLP learns a 3d weight array $\mathrm{W_r} \in \mathbb{R}^{t \times k_{max} \times k_{max}}$ and bias matrix $\mathrm{b_r} \in \mathbb{R}^{k_{max} \times k_{max}}$ to produce an initial relevance matrix $ \widetilde{r}(x) = \sigma(\mathrm{W_r}x+\mathrm{b_r})$. 
To ensure the symmetry of the relevance matrix, the second hidden layer of the relevance MLP computes $r(x) = (\widetilde{r}(x)^T + \widetilde{r}(x))/2$.

We conjecture that improvements to these models could further improve results.
First, other neural networks may be good alternatives. 
For instance, in terms of the sequential examination process, we could iteratively output the propensities $p_1(x), p_2(x), ..., p_{max}(x)$ using a recurrent neural network (RNN), where the input sequence consists of repeated context features $x$.
Second, embedding the position $k$ as a feature would give rise to different network architectures.
For example, the position could be encoded in a one-hot feature vector $\mathbf{k} \in \mathbb{R}^{k_{max}}$, which could then be concatenated to the context features $x$ to predict the examination propensity $p_k(x)$.

\section{Empirical Evaluation} \label{sec:exp}

We empirically evaluate the effectiveness and robustness of our method through real-world experiments on the ArXiv Full-Text Search \footnote{http://search.arxiv.org:8081/} and through semi-synthetic experiments on the Yahoo Learning-To-Rank Challenge corpus (set 1) \citep{chapelle2011yahoo}. The ArXiv experiments verify real-world relevance and applicability, while the synthetic experiments enable the evaluation of the method over a wide range of scenarios.

\subsection{Real-World Evaluation: ArXiv Search}

To verify that contextual effects on the propensity exist in real-world settings and to show that these can be estimated using intervention harvesting and the \algname\ estimator for the CPBM, we conducted a series of experiments on the ArXiv Full-Text Search.
To get reliable propensity estimates that can serve as a gold-standard, we fielded explicit swap intervention in addition to an A/B test that we use for intervention harvesting. 
Specifically, we assigned equal probability of accepting an incoming query to these two mechanisms.
For intervention harvesting, we used three ranking functions $\{f_1,f_2,f_3\}$ and chose uniformly at random between them for half of the incoming queries. For the other half, we also chose one of these ranking functions at random but inserted an explicit swap intervention between rank $1$ and rank $k \in \{1, 2, ..., 21\}$. These explicit swap interventions were then used to get a gold-standard estimate of the propensities via the methods in \cite{Joachims/etal/17a}. To avoid any confounding due to changes in the query distribution, data for all conditions was collected in parallel between May 14, 2018 and December 13, 2018. 
In total, 138,600 queries and 112,000 clicks were collected, with about 61,100 queries for the explicit intervention and the rest for the intervention harvesting. 
For the following experiments, the data was randomly divided into a training set with 80 \% of the data, a validation set with 10 \%, and a test set with the remaining 10 \%. In all experiments, the hyper-parameters of the neural networks in the CPBM were selected via cross-validation.

\begin{table}
\centering
\caption{Size of data sets from ArXiv.}\label{tab:size}
\vspace*{-0.3cm}
\begin{tabular}{lcccc}
\toprule
Type  & \multicolumn{2}{c}{Swap intervention} & \multicolumn{2}{c}{A/B Test Harvesting}\\
\cmidrule(lr){2-3} \cmidrule(lr){4-5}
& Clicks & Queries & Clicks & Queries\\
\midrule
Complex & 32,108 & 24,460 & 41,638 & 27,072 \\
Simple & 15,296 & 36,659 & 22,915 & 50,443 \\
\bottomrule
\end{tabular}
\vspace*{-0.1cm}
\end{table}

\paragraph{\bf Do real-world propensity curves actually depend on context?}

We first verify that the propensity curves in ArXiv do indeed depend on context. To this effect, we introduce a single binary context feature that characterizes each query as either complex (denoted by 1) or simple (denoted by 0).

Complex queries are those that contain some logical operators from the Boolean query language supported by the search engine, such as ``OR" and ``AND'', while simple queries are the remainder. The numbers of queries and clicks are given in Table~\ref{tab:size}. We then use the gold-standard propensity estimator from \cite{Joachims/etal/17a} to learn two PBM models from the swap intervention data, one for complex and one for simple queries. 

Figure~\ref{fig:complex} shows that the two propensity curves are indeed substantially different. The shaded region for each curve depicts a 95 \% confidence interval run on 1000 bootstrap samples. One possible interpretation is that complex queries are often used as more of a ``lookup" rather than a search, and thus the first few results typically either match or the user reformulates. On the other hand, simple queries are often part of an exploratory search, such that users go further down the ranking.

\begin{figure}[t]
  \centering
    \vspace*{-0.3cm}\includegraphics[width=0.80\columnwidth]{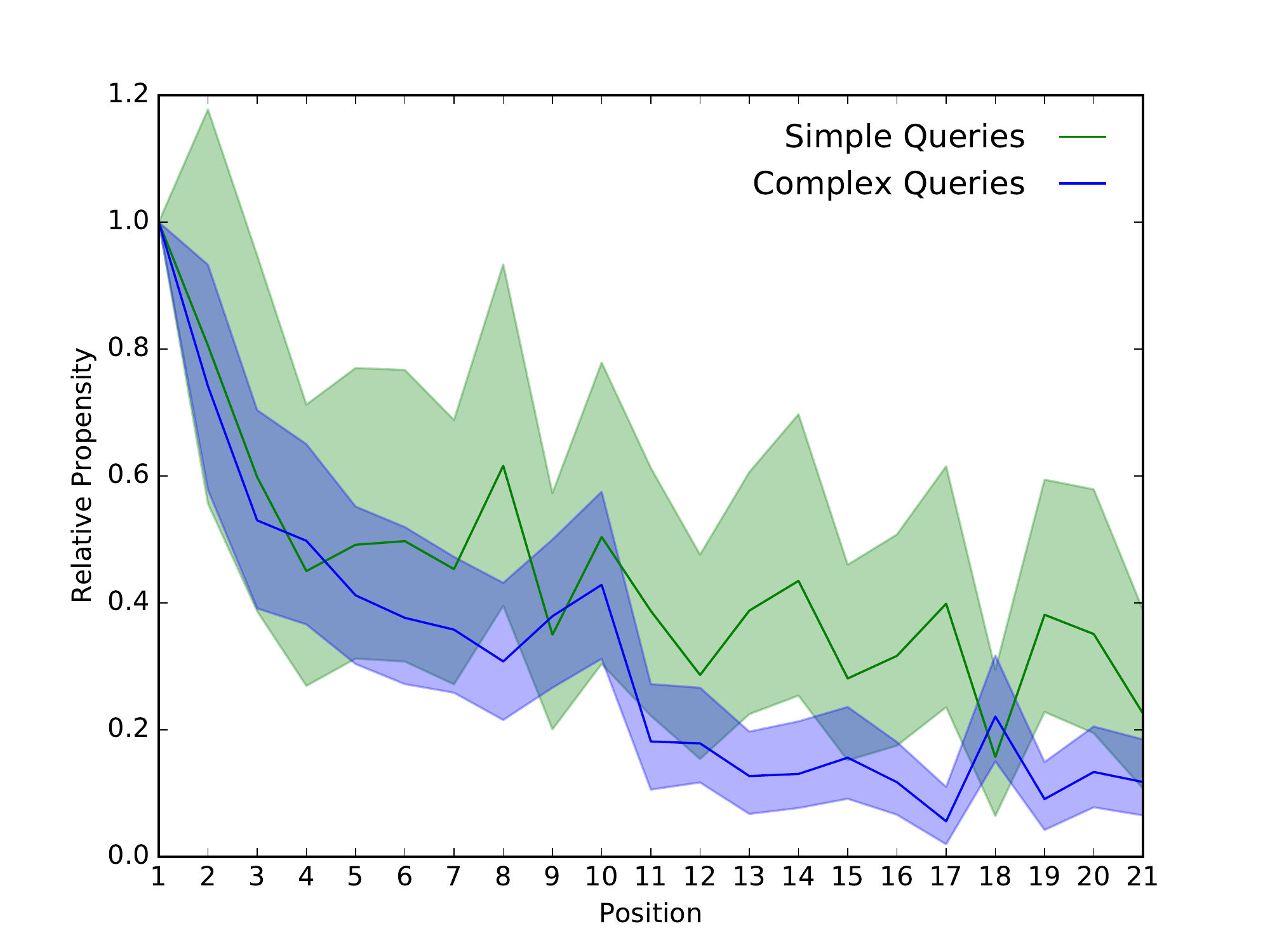}
    \vspace*{-0.4cm}
    \caption{Propensity curves for simple and complex queries on ArXiv estimated as two PBM via swap interventions.}
  \label{fig:complex}
  \vspace*{-0.2cm}
\end{figure}

\begin{figure}[t]
  \centering
    \vspace*{-0.3cm}\includegraphics[width=0.80\columnwidth]{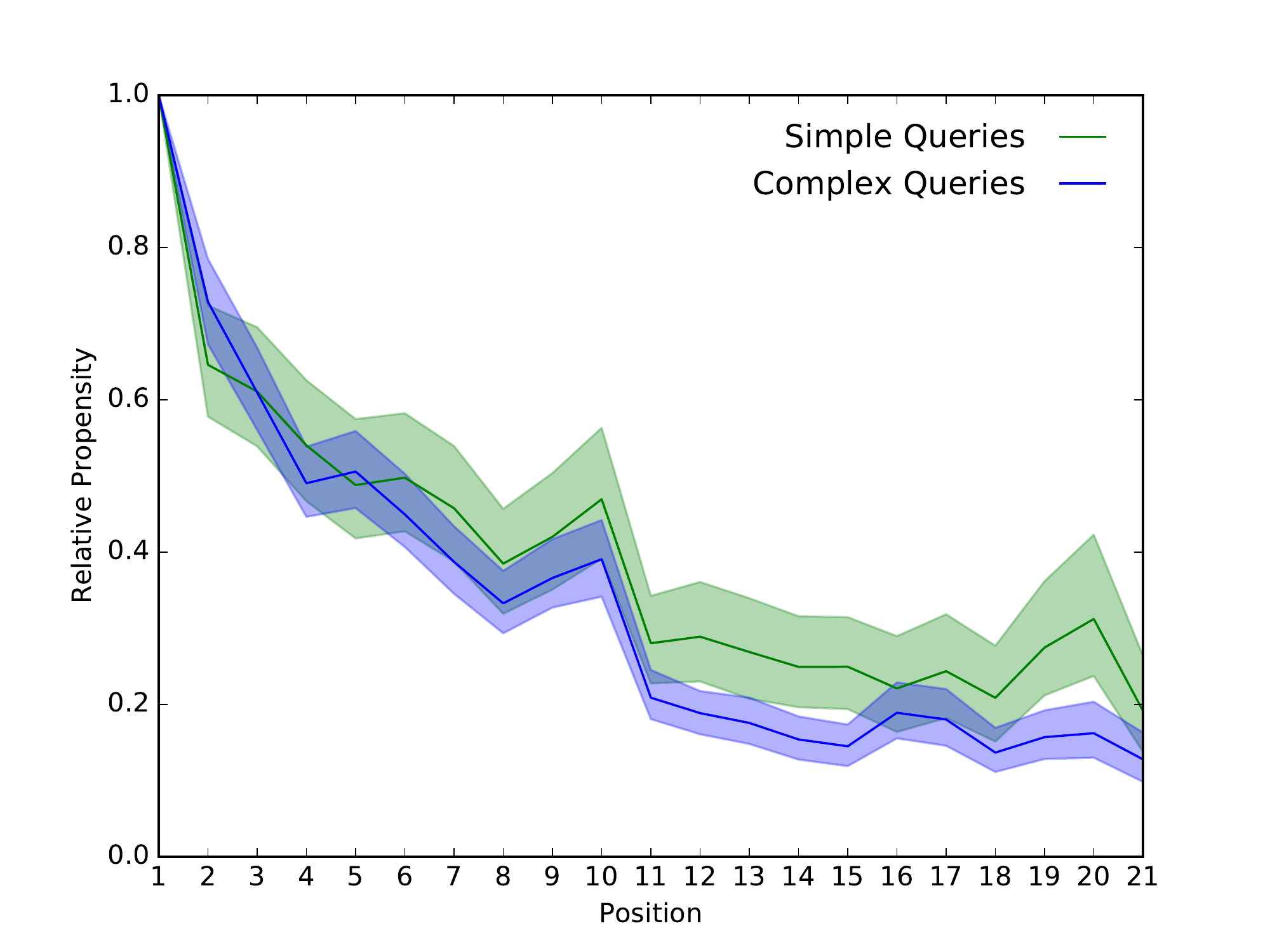}
    \vspace*{-0.4cm}
    \caption{Propensity curves for simple and complex queries on ArXiv estimated as a CPBM via intervention harvesting.}
  \label{fig:complexcpbm}
  \vspace*{-0.2cm}
\end{figure}

\paragraph{\bf Can \algname\ learn context-dependent propensity curves?}

Now that we know that contextual dependencies exist in real-world propensity curves, we can verify whether the \algname\ estimator with the neural CPBM model can accurately estimate these curves. Figure~\ref{fig:complexcpbm} show the propensity curves estimated by the \algname\ estimator on the intervention harvesting data from Table~\ref{tab:size} using the neural model with only the single input feature. The curves closely match the gold standard in Figure~\ref{fig:complex}, indicating that the CPBM can accurately learn these curves with a single neural network model. In addition, \algname\ achieves much improved error bars. This is to be expected, given that \algname\ makes more efficient use of the data than the ratio-estimates from \cite{Joachims/etal/17a}.

\begin{table*}
\caption{Objective on the test set for the PBM and the CPBM when including each feature group and for all features.}
\label{tab:max}
\vspace*{-0.3cm}
\centering
\begin{tabular}{lccccccc}
\toprule
Model & PBM & \multicolumn{6}{c}{CPBM}\\
\cmidrule(lr){3-8}
     &    & category  & query\_len & ord\_in\_session & \#results & result\_dist & All     \\
\midrule
Objective & -13926.18 & -12622.96 & -12674.8 & -13205.21 & -13241.28  & -12901.94 & -12306.52 \\
Increment (vs. PBM) & - & 1303.22 & 1251.38 & 720.97 & 684.90 & 1024.24 & 1619.66 \\
\bottomrule
\end{tabular}
\end{table*}

\paragraph{\bf Can \algname\ learn CPBM models with many context features?}

While it is infeasible to introduce additional features and learn separate PBM for each combination, adding context features to our neural CPBM model is straightforward. We will now explore in how far different groups of context features improve the predictive accuracy of the CPBM. Since we no longer have a gold-standard propensity curve to compare against, we instead use the \algname\ objective evaluated on a test set as our measure of predictive performance -- similar to evaluating log-likelihood on a test set.
We explore the following groups of context features:
\begin{enumerate}
    \item category: whether the query is specified as a category and its corresponding specified category ((binary $\{0, 1\}$, 10 features in total)
    \item query\_len: whether the length of the query is greater than $X \in \{1, 2, 5, 10, 15, 20, 25, 30, 35, 40\}$ (binary $\{0, 1\}$, 10 features in total)
    \item ord\_in\_session: whether the order of the query in its session is greater than $X \in \{1, 2, 5, 10, 15\}$ ((binary $\{0, 1\}$, 5 features in total)
    \item \#results: whether the number of results for each query is greater than $X \in \{1, 2, 5, 10, 15, 20, 50, 100, 150, 200\}$ (binary $\{0, 1\}$, 10 features in total)
    \item result\_dist: the category distribution of each query (lies in $[0,1]^{35}$ with sum to 1, 35 features in total)
\end{enumerate}
Other reasonable features can also be taken into consideration, like query performance predictors \citep{Cronen-Townsend:2002:PQP:564376.564429, amati2004query}.

Table \ref{tab:max} shows the test-set performance. The baseline is a PBM model trained according to \cite{Agarwal/etal/19}, which is essentially a CPBM model without features and a relevance model that explicitly represents each pairwise relevance. The table shows that the CPBM improves on the PBM in terms of predictiveness across all feature groups. The "category" and "query\_len" features appear to have the largest influence on the propensity curve. However, the best predictive accuracy is achieved when all features are included in the CPBM. This verifies that the CPBM can make use of complex features to improve the fit of the propensity model.




\subsection{Robustness Analysis: Yahoo LTR Challenge}

We now turn to experiments on semi-synthetic data. Using a semi-synthetic setup combines the external validity of using a real-world dataset with the ability to fully explore a range of different settings for evaluating robustness. 

Our semi-synthetic click data is based on the Yahoo LTR Challenge dataset. It contains manual relevance assessments as ground truth and we follow the given train/validation/test splits, but filter out queries that have no relevant documents. To generate click data for intervention harvesting, we learned two ranking functions by running SVM-Rank \citep{joachims2002optimizing} on two small randomly sampled subsets of the training queries. 
To control the ranker similarity, 22 queries were the same for both rankers and each ranker independently sampled 92 additional queries. 
The remaining (roughly 11,400) queries of the training set were used to generate synthetic click data based on these two ranking functions.

To generate the click data via a CPBM, we need a model for the context features and an examination model. 
For context features, each query was mapped to a 10-dimensional feature vector $x$, concatenated by two parts: relevant part $[x_1,x_2,...,x_i]$ and random part $[x_{i+1},x_{i+2},...,x_{10}]$, and we use parameter $\zeta = \frac{i}{10}$ to control the dependency between relevance and context.
For the relevant part, the important features which contribute to the relevance modeling were selected in the following way: we first used an SVM-Rank to get a one-sweep click log on the training split. 
Then we trained $k_{max}$ logistic models $r_k(x), k \in [1, k_{max}]$, which denotes the average relevance at position $k$. 
Let the coefficient of each feature $x_j$ among the given query-document feature vector in each model $r_m$ be $u_{jm}$, we assigned each feature a score $s_j = \max_j{\abs{u_{jm}}}$. 
We randomly selected $i$ features from a candidate set which contains features $x_j$ ranked in top-30 $s_j$ list.
At last, the relevant part was the average of the vector representations of all relevant results on those selected $i$ features. 
For the random part, we drew $[x_{i+1},x_{i+2},...,x_{10}]$ from the normal distribution $\mathcal{N}(0,\,\sigma^{2})$. 
To keep the performance of the PBM stable with increasing $\zeta$, $\sigma$ was tuned to be $0.35$.

For the examination model we chose $\Pr(E=1|k,x)=\frac{1}{k^{\max(w \cdot x + 1,0)}}$. The parameter vector $w$ was drawn from a uniform distribution over the half-open interval $[-\eta, \eta)$, and we normalized the weight to $\sum_{i=1}^{10} w_i = 0$ by subtracting the average weight from each position.

The parameter $\eta$ controls how much examination varies with context. In the extreme case of $\eta=0$, there is no context dependency, and context dependency grows as $\eta$ increases.
We also incorporated click noise into the simulation by setting the probability of clicking on an irrelevant result to $\epsilon_{-} = 0.1$. We chose the maximum number of positions to be $k_{max}=10$.

To evaluate the accuracy of the propensity estimates on a test sample $\mathcal{D} =\{ x^j|j\in[M]\}$, we adopted the following relative error measure where $\hat{p}_k(x) = \frac{h(k,x^j)}{h(1,x^j)}$ are the estimated relative propensities and $p_k(x) = \frac{\Pr(E=1|k,x^j)}{\Pr(E=1|1,x^j)}$ are the true relative propensities are known by construction.
\[
    RelError(h)=\frac{1}{M}\sum_{j\in \mathcal{D}}\frac{1}{k_{max}}\sum_{k=1}^{k_{max}}\abs{1-\frac{\hat{p}_k(x^j)}{p_k(x^j)}}
\]
This measure evaluates the accuracy of the estimates in terms of their use as inverse relative propensity weights, which will be their primary function. The relative error reported below is evaluated on the test set, and error bars indicate the standard deviation estimated over 6 independent runs (except in Figure~\ref{fig:rel} as described below).

In our implementation of the \algname\ estimator, the propensity model and the relevance model were both implemented by a multilayer perceptron (described in Section~\ref{sec:mlp}), whose parameters were selected via cross-validation.

\begin{table}[tb]
    \caption{Relative decrease in the relative error of CPBM vs. PBM ($\#\text{Training queries} = 113590$, $\eta = 0.5$).}
    \vspace*{-0.3cm}
    \begin{center}
    \begin{tabular}{lrrr}
    \toprule
    & PBM &CPBM & Improvement \\
    \midrule
    RelError & 0.478700 & 0.169443 & 64.60\%\\
    \bottomrule
    \end{tabular}
    \label{tab:pbm}
    \end{center}
    \vskip -0.0in
\end{table}

\paragraph{\bf How much more accurate is the CPBM compared to the PBM?}

Table~\ref{tab:pbm} shows the RelError of the CPBM and the PBM on test data, where both are trained using the \algname\ estimator using a large amount of click data for training ($113,590$ training queries). It can be thought of as the asymptotic performance of the respective model. The table shows that the CPBM improves substantially over the PBM, more than halving the error. This verifies that the \algname\ estimator can effectively learn context-dependent propensity curves from harvested interventions. Note that the CPBM had no knowledge of the true functional form of the examination model that was used to generate the clicks, but had to approximate it using the neural network model.

\paragraph{\bf Does the CPBM improve learning-to-rank performance?}

In practice, the propensities coming from the CPBM will typically be used for learning new ranking functions from the de-biased click data. We now evaluate whether the CPBM model improves learning performance compared to using the propensities from the PBM.
 
We trained a Clipped Propensity SVM-Rank \cite{Joachims/etal/17a} for each of the following three propensity models: PBM estimated via \algname, CPBM estimated via \algname, and -- as gold standard -- the true propensities used during synthetic data generation.
All hyper-parameters were picked via cross-validation.
For rank $r > 21$, we impute the propensity $p_r(x) = p_{21}(x)$.
Following \cite{Joachims/etal/17a}, we measure test-set ranking performance via the average sum of the ranks of the relevant results across the queries in the test set $\mathcal{D}$,
\begin{equation*}
    AvgRank(f) = \frac{1}{M}\sum_{j\in \mathcal{D}} \sum_{d \in \Omega(x^j)} \rank(d|f(x^j)) \rel(x^j, d).
\end{equation*}

\begin{figure}[tb]
\centering
\vspace*{-0.3cm}\includegraphics[width=0.75\columnwidth]{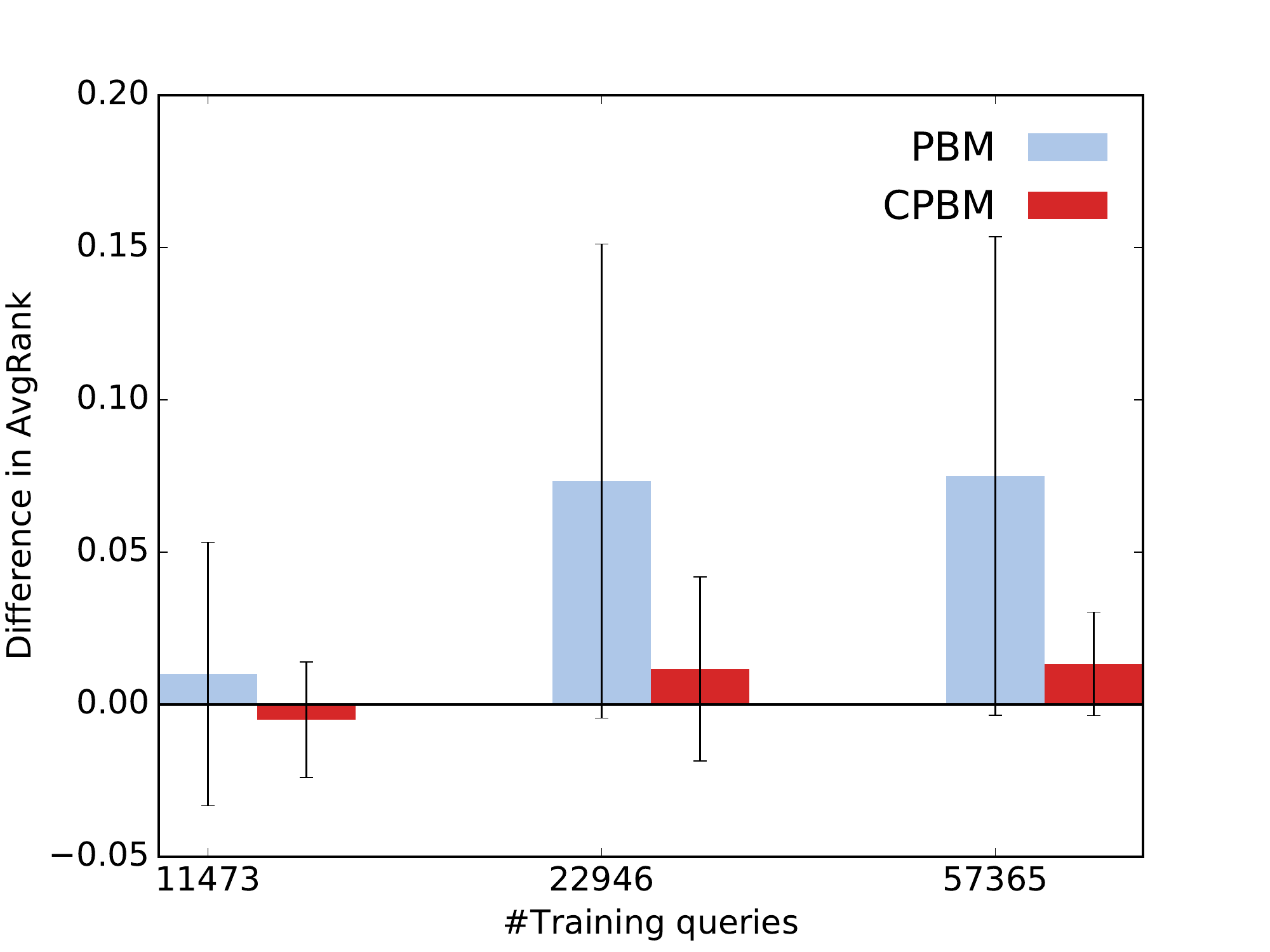}
\vspace*{-0.4cm}
\caption{Difference in AvgRank compared to the true propensity model ($\eta = 10$, $\zeta=1$).}
\label{fig:learn}
\vspace*{-0.1cm}
\end{figure}

Figure~\ref{fig:learn} shows ranking performance relative to the performance of the Propensity SVM-Rank that has access to the true propensities. For sufficiently large data set sizes, the performance when using the CPBM propensities appears closer to the gold-standard performance than when using the PBM. This is to be expected, since the training objective the Propensity SVM-Rank is known to be biased for the misspecified propensities of the PBM, so that more data no longer translates into better learning performance.

\paragraph{\bf How much data is needed to learn a \mbox{CPBM}?} 

\begin{figure}[tb]
\centering
\vspace*{-0.3cm}\includegraphics[width=0.75\columnwidth]{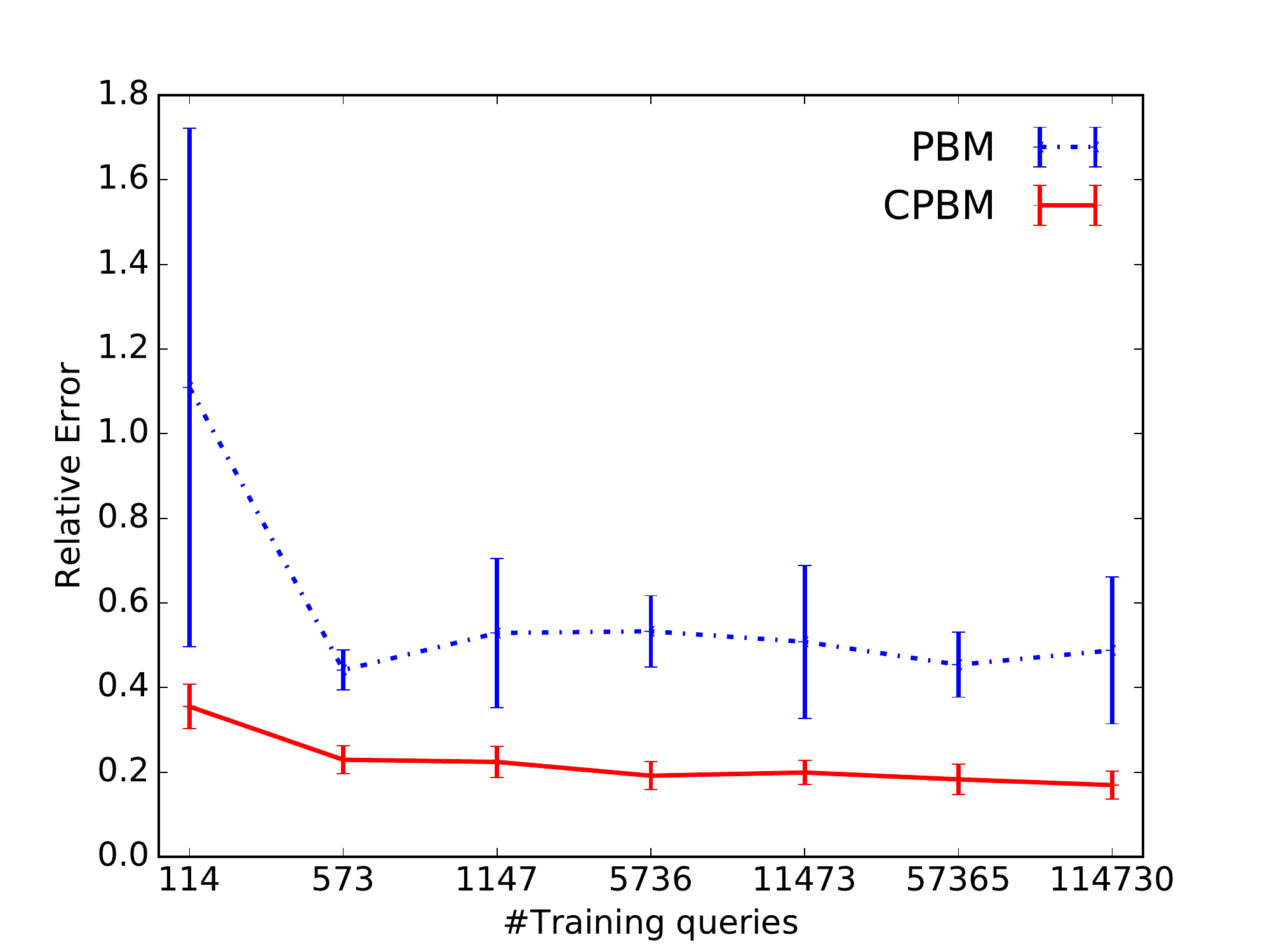}
\vspace*{-0.4cm}
\caption{Relative error with increasing number of training queries ($\eta = 0.5$, $\zeta=1$).}
\label{fig:data}
\vspace*{-0.2cm}
\end{figure}
So far, we have used large amounts of training data to study the asymptotic performance of \algname\ for the CPBM. But how much data is really needed? Figure~\ref{fig:data} compares the error of the three models across a wide range of training data sizes. The figure shows that a much smaller number of training examples suffices to get good accuracy. In particular, the relative error decreases quickly and asymptotes at about 5,700 training queries. 
Furthermore, Figure~\ref{fig:data} shows that the CPBM dominates the PBM across the whole range of data-set sizes, even when the amount of click data is quite small.

\paragraph{\bf How does the strength of context dependence affect the \mbox{CPBM}?}
\begin{figure}[tb]
\centering
\vspace*{-0.3cm}\includegraphics[width=0.75\columnwidth]{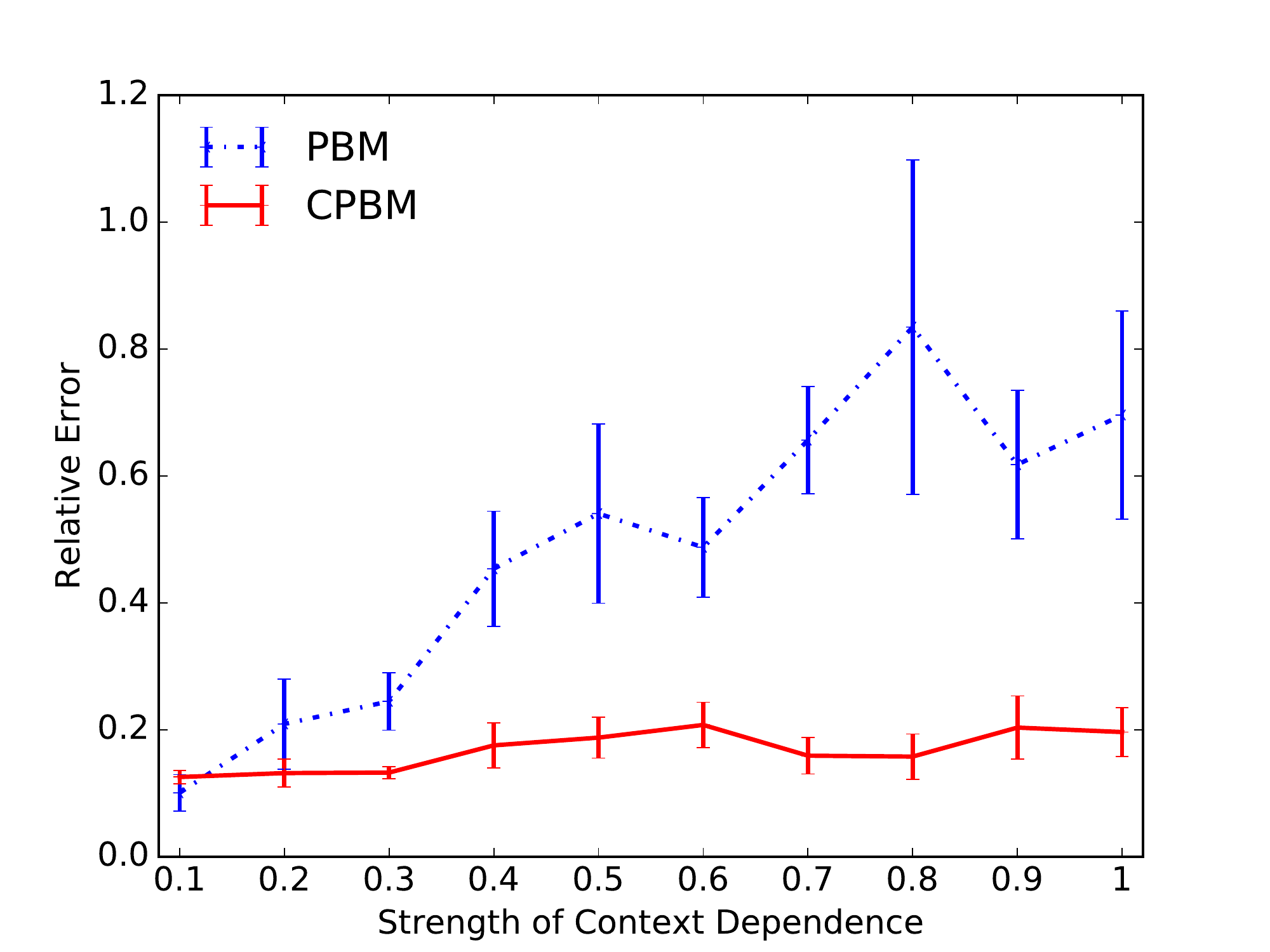}
\vspace*{-0.4cm}
\caption{Relative error with increasing strength of context dependence $\eta$ ($\#\text{Training queries} = 57365$, $\zeta=1$).}
\label{fig:ctx}
\vspace*{-0.1cm}
\end{figure}
We explore the behavior of the estimators when we vary the strength of context dependence via $\eta$. 
Results are shown in Figure~\ref{fig:ctx}, where the CPBM outperforms or at least matches the PBM across the whole range. 
As expected, the error of the PBM increases as the strength of context dependence increases. 
In contrast, the CPBM can capture the context dependence effectively.

\paragraph{\bf How important is it to incorporate a relevance model in the estimator?}
\begin{figure}[tb]
\centering
\vspace*{-0.3cm}\includegraphics[width=0.75\columnwidth]{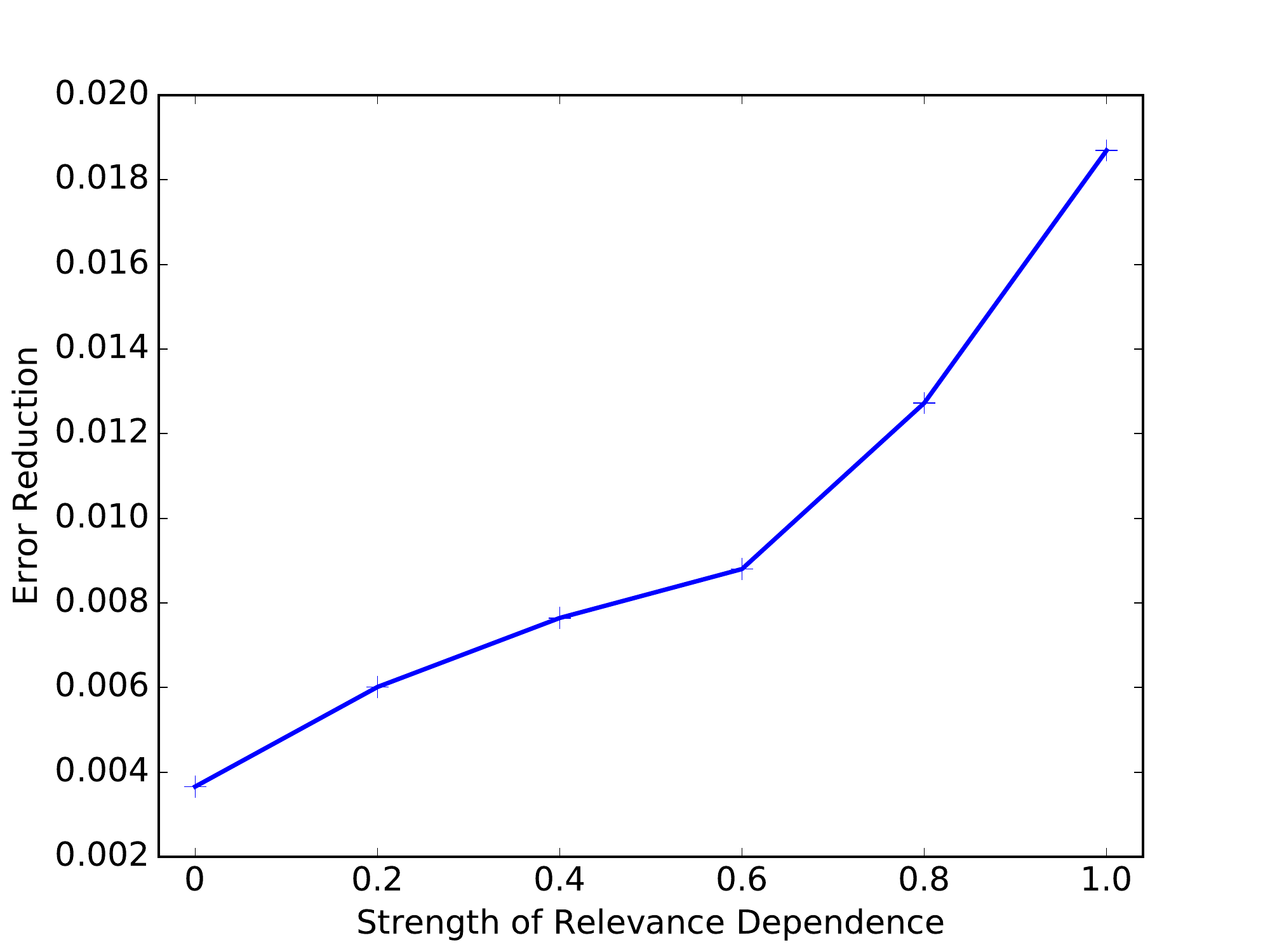}
\vspace*{-0.4cm}
\caption{Error reduction by incorporating a relevance model with increasing strength of relevance dependence $\zeta$ ($\#\text{Training queries} = 57365$, $\eta=1$).}
\label{fig:rel}
\vspace*{-0.1cm}
\end{figure}
Figure~\ref{fig:rel} shows the error reduction between the estimators under the CPBM with and without a context-dependent relevance model. 
For the CPBM with a context-dependent relevance model, we use the neural-network relevance model $g(k,k',x)$, and for the other one we simply use context-independent parameters $r_{k,k'}$ for each pair of ranks. 
To ensure statistical stability, we reran the experiment 20 times.
The error reduction provided by context-dependent relevance model increases when the context has increasing influence on the relevance profile.
With maximum decrease in error of only $0.02$, the context-dependent relevance model provides only a mild improvement to the accuracy of the estimates. This highlights the desirable fact that the relevance model $g(k,k',x)$ can be far less crucial than the query-document relevance model $g(q,d)$ in generative models.

\paragraph{\bf How accurate is the estimate at different positions in the ranking?} 
\begin{figure}[tb]
\centering
\vspace*{-0.3cm}\includegraphics[width=0.75\columnwidth]{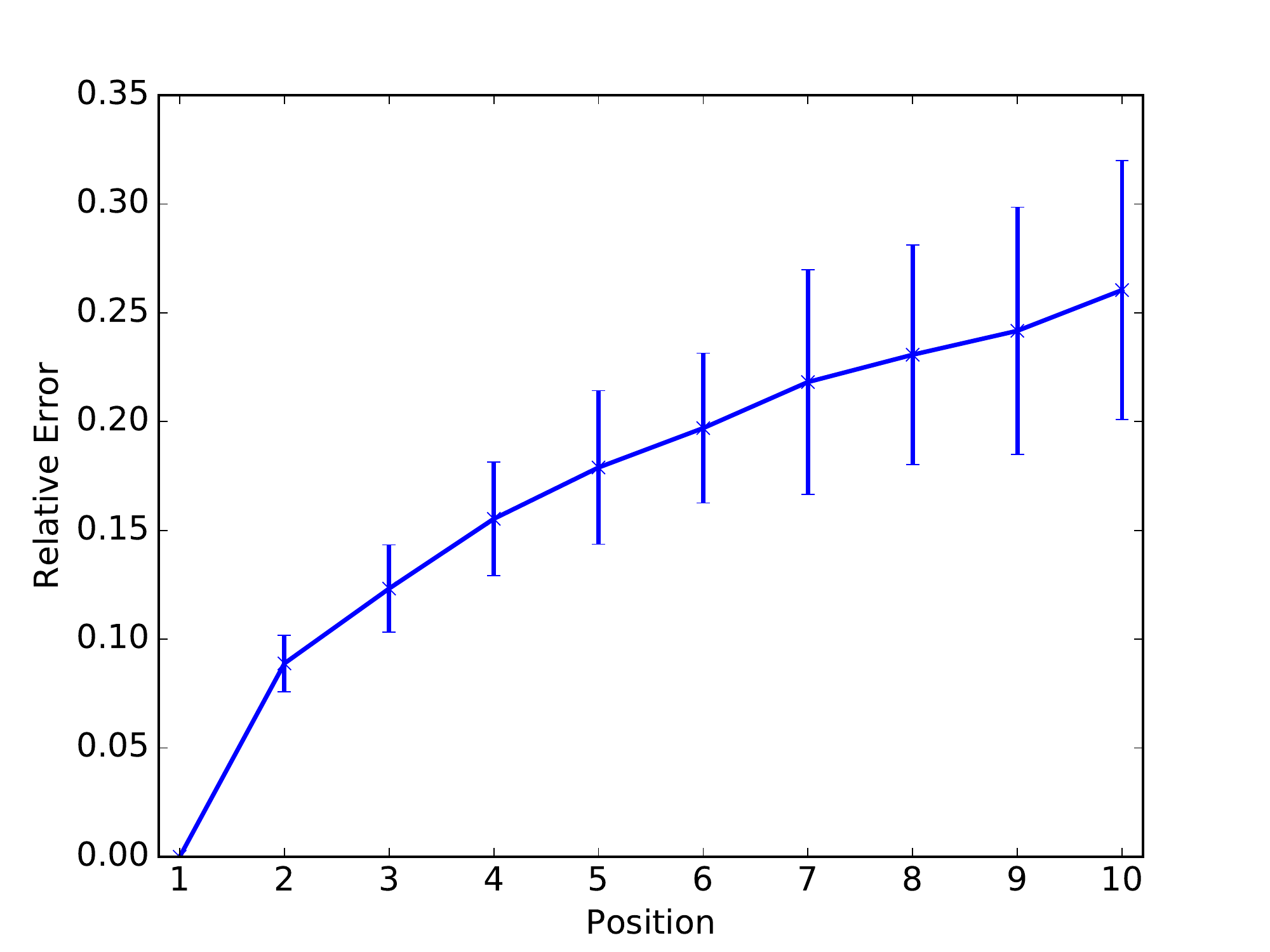}
\vspace*{-0.4cm}
\caption{Relative error of the CPBM at different positions in the ranking ($\#\text{Training queries} = 114730$, $\eta = 0.5$, $\zeta=1$).}
\label{fig:rank}
\vspace*{-0.2cm}
\end{figure}
Figure~\ref{fig:rank} shows the relative error of the CPBM at different positions in the ranking. As expected, the relative error increases with position, because lower-ranked documents have a smaller chance of receiving clicks and thus have less training data from intervention harvesting. Furthermore, the examination propensities at lower ranks are generally smaller, such that absolute deviations in the propensity estimates lead to larger contributions to our relative error metric.

\section{Conclusions}
We introduced the Contextual Position-Based Model (CPBM) to better capture the examination bias in interaction feedback from rankings.
The CPBM captures how examination changes with context, and we developed an estimator for learning a CPBM from implicit feedback data. 
The key idea is to harvest interventions from the logs of multiple historic rankers, which provides experimental control to eliminate confounding of relevance on examination. 
Plugging a neural network model into the estimator, we show how the CPBM and the estimator can effectively learn context-dependent examination models in simulation experiments and real-world experiments.

%
\begin{acks}
This research was supported in part by NSF Awards IIS-1615706 and IIS-1513692, as well as a gift from Google. All content represents the opinion of the authors, which is not necessarily shared or endorsed by their respective employers and/or sponsors.
\end{acks}

%
\bibliographystyle{ACM-Reference-Format}
\bibliography{sigir19-sigconf}

%
\end{document}